\documentclass[11pt]{article}
\usepackage{stmaryrd}
\usepackage{amsmath}
\usepackage{amssymb}
\usepackage{hyperref}
\usepackage{fullpage}
\usepackage[classfont=bold,
langfont=roman,
funcfont=italic]{complexity}
\usepackage{amsthm}
\usepackage{hyperref}
\hypersetup{colorlinks=true,linkcolor=black,citecolor=blue}
\usepackage{algpseudocode}
\usepackage{framed}
\usepackage{color}

\newenvironment{shadequote}%
{\begin{snugshade}\begin{quote}}
{\hfill\end{quote}\end{snugshade}}

\definecolor{shadecolor}{rgb}{0.9,0.9,0.9}

\newclass{\CLS}{CLS}
\newclass{\FPSPACE}{FPSPACE}
\newclass{\FSPACE}{FSPACE}
\newtheorem{theorem}{Theorem}
\newtheorem{lemma}{Lemma}

\newtheorem{definition}{Definition}
\newtheorem{problem}{Problem}

\newtheorem{claim}{Claim}
\renewcommand{\A}{\textsf{A{\scriptsize RRIVAL}}}

\renewcommand{\SAC}{\textsf{S-A{\scriptsize RRIVAL}}}

\renewcommand{\R}{\textsf{R{\scriptsize UN}}}
\newcommand{\FS}{\textsf{L{\scriptsize OCALOPT}}}
\newcommand{\FSP}{\textsf{S{\scriptsize INK}-O{\scriptsize F}-P{\scriptsize ATH}}}

\title{\textbf{Did the Train Reach its Destination: \\
The Complexity of Finding a Witness}}
\author{Karthik C.\ S.\thanks{This work was partially supported by ISF-UGC 1399/14 grant.}\\
Department of Applied Mathematics and Computer Science\\
Weizmann Institute of Science\\
\texttt{karthik.srikanta@weizmann.ac.il}}

\date{}

\begin{document}
\maketitle
\begin{abstract}
Recently, Dohrau et al.\ studied a zero-player game on switch graphs and proved that deciding the termination of the game is in $\NP\cap\coNP$. 
In this short paper, we show that the search version of this game on switch graphs, i.e., the task of finding a witness of termination (or of non-termination) is in \PLS. 
\end{abstract}
\section{Introduction}
Over the years, switch graphs have been a natural model for studying many combinatorial problems (see \cite{KRW12} and references therein). Dohrau et al.\ \cite{DGKMW16} study a problem on switch graphs, which as they suggest fits well in the theory of cellular automata. 
Informally, they describe their problem in the following way.
\begin{shadequote}
Suppose that a train is running along a railway network, starting from a designated origin, with the goal of reaching a designated destination. The network, however, is of a special nature: every time the train traverses a switch, the switch will change its position immediately afterwards. Hence, the next time the train traverses the same switch, the other direction will be taken, so that directions alternate with each traversal of the switch.

Given a network with origin and destination, what is the complexity of deciding whether the train, starting at the origin, will eventually reach the destination?
\end{shadequote}

They showed that deciding the above problem lies in $\NP\cap\coNP$.

In this paper, we address the complexity of the search version of the above problem. From a result of Megiddo and Papadimitriou \cite{MP91}, we have that     $F(\NP\cap\coNP)\subseteq\TFNP$, i.e., the search version of any decision problem in $\NP\cap\coNP$ is in \TFNP. We show that the search version of the problem considered by Dohrau et al.\ is in \PLS, a complexity class inside \TFNP\ that captures the difficulty of finding a locally optimal solution in optimization problems. 

\section{Preliminaries}
We use the following notation $[n]=\{1,\ldots ,n\}$ and $\llbracket n\rrbracket=\{0,\ldots ,n\}$. We recapitulate here the definition of the complexity class \PLS , introduced by Johnson et al. \cite{JPY88}. There are many equivalent ways to define the class \PLS\ and below we define it through its complete problem \FS\ similar to \cite{DP11}.

\begin{definition}[\FS]
Given circuits $S:\{0,1\}^n\to\{0,1\}^n$, and $V:\{0,1\}^n\to [2^n]$, find a string $x\in\{0,1\}^n$ satisfying $V(x)\ge V(S(x))$.
\end{definition}

\begin{definition}
\PLS\ is the class of all search problems which are polynomial time reducible to \FS.
\end{definition}

Below, we recollect some definitions introduced in \cite{DGKMW16}.

\subsection{\A\ Problem}

We start with the object of study: switch graphs. We use the exact same notations as in \cite{DGKMW16}.

\begin{definition}[Switch Graph]
A switch graph is a 4-tuple $G=(V,E,s_0,s_1)$, where $s_0,s_1:V\to V$, $E=\{(v,s_0(v))\mid v\in V\}\cup\{(v,s_1(v))\mid v\in V\}$, with self-loops allowed. For every vertex $v\in V$, we refer to $s_0(v)$ as the even successor of $v$, and we refer to $s_1(v)$ as the odd successor of $v$. For every $v\in V$, $E^+(v)$ denotes the set of outgoing edges from $v$, and $E^-(v)$ denotes the set of incoming edges to $v$.
\end{definition}

Next, consider a formal definition of the procedure \R, which captures the run of the train described in the introduction.

\begin{definition}[\R\ Procedure given in \cite{DGKMW16}]
Given a switch graph $G=(V,E,s_0,s_1)$ with origin and destination $o,d\in V$, the procedure \R\ is described below.
For the procedure, we assume arrays \verb+s_curr+ and \verb+s_next+, indexed by $V$, such that initially \verb+s_curr+$[v] =
s_0(v)$ and \verb+s_next+$[v] = s_1(v)$ for all $v \in V$.

\begin{algorithmic}
\Procedure{\R}{G,o,d}
\State $v:=o$
\While{$v\neq d$}
\State $w:=$\verb+s_curr+$[v]$
\State swap (\verb+s_curr+$[v]$, \verb+s_next+$[v]$)
\State $v := w$ \hfill $\rhd$ traverse edge $(v,w)$
\EndWhile
\EndProcedure
\end{algorithmic}
\end{definition}

The problem \A, considered in \cite{DGKMW16} is the following.

\begin{problem}[\A]
Given a switch graph $G=(V,E,s_0,s_1)$, an origin $o\in V$, and a destination $d\in V$, the problem \A\ is to decide if the procedure \R\ terminates or not.
\end{problem}

\begin{theorem}[\cite{DGKMW16}]\label{Base}
\A\ is in $\NP \cap\coNP$.
\end{theorem}

In order to prove the above result, the authors consider the run profile as a witness. Elaborating, the run profile is a function which assigns 
to each edge the number of times it has been traversed during the procedure \R. It is easy to note that a run profile has
to be a switching flow.

\begin{definition}[Switching Flow, as defined in \cite{DGKMW16}]\label{switch}
Let $G=(V,E,s_0,s_1)$ be a switch graph, and let $o, d \in V$, $o \neq d$. A switching flow is a function \emph{\textbf{x}} $: E \to \mathbb{N}_0$ (where \emph{\textbf{x}}$(e)$ is denoted as $x_e$) such that the following two conditions hold for all $v\in V$.
\begin{align}
&\sum_{e\in E^+(v)}x_e -\sum_{e\in E^-(v)}x_e=\begin{cases}1,&\ v=o,\\
-1,&\ v=d,\\
0,&\ \text{otherwise.}
\end{cases}\\
&0\le x_{(v,s_1(v))}\le x_{(v,s_0(v))}\le x_{(v,s_1(v))}+1. 
\end{align}
\end{definition}

\sloppypar{Note that while every run profile is a switching flow, the converse is not always true  as the balancing condition (2) fails to capture the strict alternation between even and odd successors. Nonetheless, the existence of a switching flow implies the termination of the \R\ procedure (Lemma~1~of~\cite{DGKMW16}). 
%



\section{\SAC\ Problem}\label{search}

Now, we describe a reduction from an instance of \A\ to two instances of \A\ (this is an implicit step in the proof of Theorem~\ref{Base}). 
Given an instance $(G,o,d)$ of \A, we build two new instances of \A, $\left(H,\overline{o},d\right)$ and $(H,\overline{o},\overline{d})$, where $H=(V\cup\{\overline{o},\overline{d}\},E^\prime,s_0^\prime,s_1^\prime)$ is a switch graph specified below. 
Let $X_d$ be the following subset of the vertex set of $G$:
$$X_d=\left\{v\mid \text{There is no directed path in $G$ from }v\text{ to }d\right\}.$$
The vertex set of $H$ is the vertex set of $G$ with the addition of two new vertices $\overline{o}$ and $\overline{d}$. 
We define $s_0(\overline{o})=s_1(\overline{o})=o$.
For $i\in\{0,1\}$ and $v\in V\cup\{\overline{d}\}$, we have that $s_i^\prime$ of $H$ is obtained from $s_i$ of $G$  as follows.
\begin{align*}
s_i^\prime(v)=\begin{cases}
v,&\ v\in\{d,\overline{d}\},\\
\overline{d},&\ v\in X_d,\\
s_i(v),&\ \text{otherwise.}
\end{cases}
\end{align*}

This reduction has the following property. 

\begin{claim}\label{ReductiontoStructure}
If $(G,o,d)$ is an YES instance of \A\ then, $(H,\overline{o},d)$ is an YES instance of \A\ and $(H,\overline{o},\overline{d})$ is a NO instance of \A. On the other hand, 
if $(G,o,d)$ is a NO instance of \A\ then, $(H,\overline{o},d)$ is a NO instance of \A\ and $(H,\overline{o},\overline{d})$ is an YES instance of \A.
\end{claim}

The proof of the above claim follows from the proof of Theorem~3 in \cite{DGKMW16}. 
We are now ready to describe a search version of the \A\ problem. 

\begin{problem}[\SAC]
Given a switch graph $G=(V,E,s_0,s_1)$, an origin $o\in V$, and a destination $d\in V$, the problem \SAC\ is to either find a switching flow of $(H,\overline{o},d)$ or a switching flow of $(H,\overline{o},\overline{d})$.
\end{problem}

We have that from Lemma~1~of~\cite{DGKMW16}, a switching flow of $(H,\overline{o},d)$ is an \NP-witness for the existence of a run profile of $(G,o,d)$, and that a switching flow of $(H,\overline{o},\overline{d})$ is a \coNP-witness for the non-existence of a run profile of $(G,o,d)$. Thus, \SAC\ is the appropriate search version problem of the \A\ problem.
From Claim~\ref{ReductiontoStructure}, \SAC\ is clearly in \TFNP. In the next section, we show that \SAC\ is in \PLS, a subclass of \TFNP. 

Below, we essentially show that switching flows are bounded, and this is a critical result in order to establish the reduction in Section~\ref{PLS}. Note that this is a strengthening of Lemma~2 in \cite{DGKMW16} which provided a bound on the run profile. 

\begin{lemma}\label{MainBound}
Let $G=(V,E,s_0,s_1)$ be a switch graph, $o\in V$ a origin, $d\in V$ a destination. Let $\mathbf{x}$ be a switching flow of $(H,\overline{o},u)$ for some vertex $u(\neq\overline{o})$ of $H$. Then, we have that for all $v\in (V\cup\{\overline{o}\})\setminus\{d\}$ and $i\in\{0,1\}$, the following bound holds:
$$ x_{(v,s_i(v))}< 2^n,$$
where $n$ is the number of vertices of $H$.
\end{lemma}
\begin{proof}
We recapitulate here that $H=(V\cup\{\overline{o},\overline{d}\},E^\prime,s_0^\prime,s_1^\prime)$.
First we observe that without loss of generality we can assume that $x_{(d,s_0'(d))}=x_{(d,s_1'(d))}=x_{(\overline{d},s_0'(\overline{d}))}=x_{(\overline{d},s_1'(\overline{d}))}=0$. This is because we can consider a new switching flow $\mathbf{x}'$ of $(H,\overline{o},u)$, defined such that for all $e\in E'\setminus (E^+(d)\cup E^+(\overline{d}))$, we have $x_e'=x_e$, and for all $e\in E^+(d)\cup E^+(\overline{d})$, we have $x_e'=0$. Note that the edges in $E^+(d)\cup E^+(\overline{d})$ are self loops and thus $\mathbf{x}'$ satisfies the conditions of Definition~\ref{switch}. Moreover, if the above lemma is true for $\mathbf{x}'$ then it is true for $\mathbf{x}$ as well because for all $v\in (V\cup\{\overline{o}\})\setminus\{d\}$, we have $x_{(v,s_i(v))}=x_{(v,s_i(v))}'$. 

Next, we build a switch graph $I=(V\cup\{\overline{o},\overline{d}\},E',s_0'',s_1'')$ from $H=(V\cup\{\overline{o},\overline{d}\},E',s_0',s_1')$, as follows. For all $i\in\{0,1\}$, and for all $v\in V\cup\{\overline{o},\overline{d}\}$, we define $s_i''(v)$ as follows: 
$$
s_i''(v)=\begin{cases}
s_{1-i}'(v)&\text{ if }{x}_{(v,s_0'(v))}-{x}_{(v,s_1'(v))}=1,\\
s_i'(v)&\text{ if }{x}_{(v,s_0'(v))}-{x}_{(v,s_1'(v))}=0
\end{cases}.	
$$

Let $\mathbf{y}$ be the run profile of the procedure \R\ on the switch graph $I$ starting from $u$ and ending at a vertex in $\{d,\overline{d}\}$ (if $u\in\{d,\overline{d}\}$ then, we define $\mathbf{y}=\vec{0}$). Without loss of generality we assume that the above procedure ends on vertex $d$. It is important to note that if $u\notin\{d,\overline{d}\}$ there is exactly one incoming edge $e$ of $d$ such that $y_e=1$, and for every other incoming edge $e'$ of $d$, we have $y_{e'}=0$. Let $\mathbf{z}=\mathbf{x}+\mathbf{y}$, where the addition is done point-wise. We claim that $\mathbf{z}$ is a switching flow of $(H,\overline{o},d)$. This is clearly true when $u\in\{d,\overline{d}\}$, as we have $\mathbf{z}=\mathbf{x}$. If $u\notin\{d,\overline{d}\}$, $\mathbf{z}$ is a switching flow because of the following:

\begin{itemize}
\item For every vertex $v\in (V\cup\{\overline{d}\})\setminus \{u,d\}$, we note the following: 
\begin{align*}
\sum_{e\in E^+(v)}z_e&=\sum_{e\in E^+(v)}(x_e +y_e)=
 \sum_{e\in E^+(v)}x_e+\sum_{e\in E^+(v)}y_e\\
&=\sum_{e\in E^-(v)}x_e+\sum_{e\in E^-(v)}y_e
=\sum_{e\in E^-(v)}z_e
\end{align*}

For the vertex $u$, we note the following,
\begin{align*}
\sum_{e\in E^+(u)}z_e&=\sum_{e\in E^+(u)}(x_e +y_e)=
 \sum_{e\in E^+(u)}x_e+\sum_{e\in E^+(u)}y_e\\
&=\left(-1+\sum_{e\in E^-(u)}x_e\right)+\left(1+\sum_{e\in E^-(u)}y_e\right)
=\sum_{e\in E^-(u)}z_e
\end{align*}

For the vertex $\overline{o}$, we note the following,
\begin{align*}
\sum_{e\in E^+(\overline{o})}z_e&=\sum_{e\in E^+(\overline{o})}(x_e +y_e)=
 \sum_{e\in E^+(\overline{o})}x_e+\sum_{e\in E^+(\overline{o})}y_e\\
&=\left(1+\sum_{e\in E^-(\overline{o})}x_e\right)+\left(\sum_{e\in E^-(\overline{o})}y_e\right)
=\left(\sum_{e\in E^-(d)}z_e\right)+1
\end{align*}

For the vertex $d$, we note the following,
\begin{align*}
\sum_{e\in E^+(d)}z_e&=\sum_{e\in E^+(d)}(x_e +y_e)=
 \sum_{e\in E^+(d)}x_e+\sum_{e\in E^+(d)}y_e\\
&=\left(\sum_{e\in E^-(d)}x_e\right)+\left(-1+\sum_{e\in E^-(d)}y_e\right)
=\left(\sum_{e\in E^-(\overline{o})}z_e\right)-1
\end{align*}

\item For every $v\in V\cup \{\overline{o},\overline{d}\}$, we note that if ${x}_{(v,s_0'(v))}-{x}_{(v,s_1'(v))}=1$ then we have that ${y}_{(v,s_0'(v))}-{y}_{(v,s_1'(v))}\in \{-1,0\}$, and thus consequently, ${z}_{(v,s_0'(v))}-{z}_{(v,s_1'(v))}\in\{0,1\}$. On the other hand if ${x}_{(v,s_0'(v))}-{x}_{(v,s_1'(v))}=0$ then we have that ${y}_{(v,s_0'(v))}-{y}_{(v,s_1'(v))}\in \{0,1\}$, and thus, ${z}_{(v,s_0'(v))}-{z}_{(v,s_1'(v))}\in\{0,1\}$. 
\end{itemize}
Therefore, for all $v\in V\cup \{\overline{o},\overline{d}\}$, we have that the conditions in Definition~\ref{switch} are satisfied and thus $\mathbf{z}$ is a switching flow of $(H,\overline{o},d)$. Next, we show that for all $v\in (V\cup\{\overline{o}\})\setminus\{d\}$ and $i\in\{0,1\}$, the following bound holds:
$$ z_{(v,s_i(v))}< 2^n,$$
and the lemma follows as $\mathbf{z}=\mathbf{x}+\mathbf{y}$ and $y_e\ge 0$ for all edges $e\in E'$. In order to prove the above bound on $\mathbf{z}$, we recall the definition of `desperation' of an edge from \cite{DGKMW16}. For an edge $e=(v,w)$, the length of the shortest directed path from its head $w$ to $d$ is called its desperation. Also, recollect that $X_d$ is the set of all vertices in $G$ such that there is no  directed path in $G$ from $v$ to $d$. For every $i\in\{0,1\}$ and $v\in X_d$, we have $z_{(v,s_0'(v))}=z_{(v,s_1'(v))}=0$ because $z_{(\overline{d},s_0'(\overline{d}))}=z_{(\overline{d},s_1'(\overline{d}))}=0$. This implies that  $\underset{e\in E^+(v)}{\sum}z_e=0$, and thus we have for all vertices $v'$ such that $s_i'(v')=v$ for some $i\in\{0,1\}$, we have $z_{(v',s_i'(v'))}=0$. Additionally, we note that $\underset{e\in E^-(\overline{o})}{\sum}z_e=0$ and thus we have for all $i\in\{0,1\}$ that $z_{(\overline{o},s_i'(\overline{o}))}\le 1$.

Next, we note that for every vertex $v$ in $V\setminus (X_d\cup\{d\})$ and $i\in\{0,1\}$ such that $s_i'(v)\in V\setminus X_d$, we have that the desperation of $(v,s_i'(v))$ is well-defined and less than $n$. We show by induction on the desperation value that $z_{(v,s_i'(v))}\le 2^{k+1}-1$, where $k$ is the desperation of $(v,s_i'(v))$. If $k=0$, then $s_i'(v)=d$. But since $\underset{e\in E^+(d)}{\sum}z_e=0$, we have $\underset{e\in E^-(d)}{\sum}z_e=1$, and thus $z_{(v,s_i'(v))}\le 1$. We suppose now that $k>0$, and we have from induction hypothesis that $\underset{e\in E^+(s_i'(v))}{\sum}z_e\le 2(2^k -1)+1=2^{k+1}-1$. Since $\underset{e\in E^-(s_i'(v))}{\sum}z_e=\underset{e\in E^+(s_i'(v))}{\sum}z_e$, we have $z_{(v,s_i'(v))}\le \underset{e\in E^-(s_i'(v))}{\sum}z_e\le 2^{k+1}-1$. This completes the claim that $z_{(v,s_i'(v))}\le 2^{k+1}-1$. Finally, by noting that $k<n$, we have $z_{(v,s_i'(v))}\le 2^{k+1}-1 \le 2^n-1<2^n$.
\end{proof}

\section{\PLS\ Membership}\label{PLS}
In this section, we show that $\SAC$ is in \PLS.
\begin{theorem}\label{PLSmem}
\SAC\ is in \PLS.
\end{theorem}
\begin{proof}
Given an instance $(G,o,d)$ of \SAC, we build an instance $(S,V)$ of \FS\ as follows. Let $n$ be the size of the vertex set of $H=(V^\prime,E^\prime,s_0^\prime,s_1^\prime)$ (see section~\ref{search} for definition).
We construct $S:[n]\times \llbracket 2^n\rrbracket^{2n}\to [n]\times \llbracket 2^n\rrbracket^{2n}$ and $V:[n]\times \llbracket 2^n\rrbracket^{2n}\to \llbracket 2n\cdot 2^n\rrbracket\cup\{-1\}$. Informally, an element of the domain of $V$ or $S$ is the concatenation of the label of a vertex $v$ in $H$ and a potential switching flow for $(H,\overline{o},v)$. If the switching flow for $(H,\overline{o},v)$ satisfies the two conditions of Definition~\ref{switch}, $S$ computes the switching flow of $(H,\overline{o},N(v))$, where $N(v)$ is the next vertex encountered by the procedure \R\ (assuming the switching flow for $(H,\overline{o},v)$). Similarly, if the switching flow for $(H,\overline{o},v)$ satisfies the two conditions of Definition~\ref{switch}, $V$ is the total number of times the edges in $H$ have been traversed by \R\ according to the switching flow of $(H,\overline{o},v)$. More formally, for $(v,\mathbf{x})\in [n]\times \llbracket 2^n\rrbracket^{2n}$, we construct $S((v,\mathbf{x}))$ as follows:

\begin{itemize}
\item If $\mathbf{x}$ does not satisfy the switching flow conditions for $(H,\overline{o},v)$ then, $S((v,\mathbf{x}))=(\overline{o},0^{2n})$.
\item If $\mathbf{x}$ is a switching flow of $(H,\overline{o},v)$ and $v\notin \{d,\overline{d}\}$, let $e=(v,s_i^\prime(v))$, where $i=x_{(v,s_0^\prime(v))}-x_{(v,s_1^\prime(v))}$. Then, we define $S((v,\mathbf{x}))=
(s_i^\prime(v),\mathbf{x}+\xi_{e})$, where  $\xi_e$ is a vector in $\llbracket 2^n\rrbracket^{2n}$ which is 1 on the $e^{\text{th}}$ coordinate and 0 everywhere else. Note that from Lemma~\ref{MainBound}, we have that $x_e<2^n$ and thus we have $\mathbf{x}+\xi_{e}$ on the $e^{\text{th}}$ coordinate is at most $2^n$, i.e., in the range of the output of the circuit $S$.
\item If $\mathbf{x}$ is a switching flow of $(H,\overline{o},v)$ and $v\in \{d,\overline{d}\}$ then, $S((v,\mathbf{x}))=(\overline{o},0^{2n})$.
\end{itemize}

We note here that the construction of $S$ only depends on checking if $\mathbf{x}$ is a switching flow (can be performed in $\text{poly}(n)$ time)  and finding the appropriate neighbor in $H$ (can be computed in $O(n)$ time).
Next, we construct $V((v,\mathbf{x}))$ as follows:
\begin{itemize}
\item If $\mathbf{x}$ is not a switching flow of $(H,\overline{o},v)$ then, $V((v,\mathbf{x}))=-1$.
\item If $\mathbf{x}$ is a switching flow of $(H,\overline{o},v)$ then, $V((v,\mathbf{x}))=\|\mathbf{x}\|_1=\underset{e\in E^\prime}{\sum}x_e$.
\end{itemize}

\sloppypar{
Let $(v,\mathbf{x})\in [n]\times \llbracket 2^n\rrbracket^{2n}$ be a solution to the \FS\ instance ($V,S$), i.e., $V((v,\mathbf{x}))\ge~V(S((v,\mathbf{x})))$. Suppose $\mathbf{x}$ is not a switching flow of $(H,\overline{o},v)$ then, $V((v,\mathbf{x}))=-1$ and $V(S((v,\mathbf{x})))=V((\overline{o},0^{2n}))=0$, thus $(v,\mathbf{x})$ cannot be a solution to \FS\ in that case. Suppose $\mathbf{x}$  is a switching flow of $(H,\overline{o},v)$ and $v\notin \{d,\overline{d}\}$ then, we note that $\mathbf{x}+\xi_{(v,s_i^\prime(v))}$ is a switching flow of $(H,\overline{o}, s_i'(v))$, where $i=x_{(v,s_0^\prime(v))}-x_{(v,s_1^\prime(v))}$. Thus, we have $V((v,\mathbf{x}))=\|\mathbf{x}\|_1$ and $V(S((v,\mathbf{x})))=V((s_i^\prime(v),\mathbf{x}+\xi_{(v,s_i^\prime(v))}))=\|\mathbf{x}+\xi_{(v,s_i^\prime(v))}\|_1=\|\mathbf{x}\|_1+1$, thus $(v,\mathbf{x})$ cannot be a solution to \FS\ in that case as well. This means that $(v,\mathbf{x})$ can be a solution of \FS\ only if $\mathbf{x}$ is a switching flow of $(H,\overline{o},v)$ and $v\in \{d,\overline{d}\}$.}  

Suppose there is a run profile $\mathbf{x}$ of $(H,\overline{o},d)$. From Lemma~2 in \cite{DGKMW16}, we have that $x_e<2^n$ for all $e\in E'$, and thus $S((d,\mathbf{x}))$ and $V((d,\mathbf{x}))$ are well defined. Since $\overline{o}\neq d$, we have that $\|{\mathbf{x}}\|_1>0$. Therefore,  we have that $V((d,{\mathbf{x}}))>0$. On the other hand, we have $S((d,{\mathbf{x}}))=(\overline{o},0^{2n})$. This implies that $0=V(S((d,{\mathbf{x}})))< V((d,{\mathbf{x}}))$. Therefore $(d,{\mathbf{x}})$ is a solution of the instance $(S,V)$ of \FS.  

Suppose there is a run profile $\mathbf{x}$ of $(H,\overline{o},\overline{d})$. Again from Lemma~2 in \cite{DGKMW16}, we have that $x_e<2^n$ for all $e\in E'$, and thus $S((\overline{d},\mathbf{x}))$ and $V((\overline{d},\mathbf{x}))$ are well defined. Since $\overline{o}\neq \overline{d}$, we have that $\|\mathbf{x}\|_1>0$. Therefore,  we have that $V((\overline{d},\mathbf{x}))>0$. On the other hand, we have $S((\overline{d},\mathbf{x}))=(\overline{o},0^{2n})$. This implies that $0=V(S((\overline{d},\mathbf{x})))< V((\overline{d},\mathbf{x}))$. Therefore $(\overline{d},\mathbf{x})$ is a solution of the instance $(S,V)$ of \FS. 
\end{proof}

\section{Discussion and Conclusion}
In this short paper, we have introduced a search version of the problem \A, called \SAC, and showed that \SAC\ is contained in \PLS. The main  open problem is to determine the hardness of \SAC\ (or equivalently \A): is \SAC\ in \FP\ or is it \PLS-hard? 

Additionally, one could try to address the complexity of finding the run profile of $(G,o,d)$, i.e., determining the number of times each edge has been traversed during the procedure \R. We suspect that this might be \FPSPACE-complete\ because the following related problem can be shown to be \FPSPACE-complete.  

\begin{definition}[\FSP]
Given circuits $S:\{0,1\}^n\to\{0,1\}^n$ and $V:\{0,1\}^n\to \llbracket 2^n\rrbracket$, and $s^\star\in\{0,1\}^n$, find a string $x\in\{0,1\}^n$ such that there is some integer $r\in \llbracket 2^n\rrbracket$ satisfying $S^r(s^\star)=x$ and $V(x)\ge V(S(x))$.
\end{definition}

\begin{theorem}[Similar to Theorem 2 of \cite{P94}]\sloppypar{
\FSP\ is \FPSPACE-complete.}
\end{theorem}
\begin{proof}
We shall show that \FSP\ is \FPSPACE-hard, as the membership is immediate. Fix any problem $\Pi$ in $\FPSPACE$. By definition of $\FPSPACE$ there exists some polynomial $p:\mathbb{N}\to\mathbb N$ and a Turing machine $T$ on alphabet $\Sigma$ which solves $\Pi$ for every input of size $n$ using at most $p(n)$ space. Starting from an instance $I$ of $\Pi$ (of size $n$) we build an instance of \FSP\ as follows. Let $N=\left|\Sigma\right|^{p(n)}$. We fix $s^\star$ to be the configuration of the Turing machine with only the input concatenated with a counter set to 0. For each configuration $c$ of $T$ and counter value $i$, we set the computation of $S$ to be equal to the configuration of $T$ after running one step starting from $c$, concatenated with the counter incremented to $i+1$ modulo $N$. We set the computation of $V$ to be equal to the counter modulo $N$. 
Note that after $T$ with input $I$ has halted the configuration of $T$ doesn't change and also that $T$ halts in at most $N$ steps (as no configuration can be repeated). The reduction to \FSP\ follows.
\end{proof}

Therefore, if \SAC\ was \PLS-hard, one could plausibly utilize the reduction from \FS\ to \SAC\ to show that determining the run profile is \FPSPACE-complete, as finding the sink of a given path for the \FS\ problem is \FPSPACE-complete.

\subsection*{Acknowledgment}
I would like to thank Bernd G\"{a}rtner, Oded Goldreich, and Eylon Yogev for several helpful discussions.  I would also like to thank the anonymous reviewers for their helpful comments. 

\bibliographystyle{alpha}
\bibliography{geom}
\end{document}